%% file: main.tex
\pdftrailerid{}
\documentclass[10pt,twocolumn]{article}

\usepackage[margin=0.72in,columnsep=0.24in]{geometry}
\usepackage[T1]{fontenc}
\usepackage[utf8]{inputenc}
\usepackage{lmodern}
\usepackage{microtype}
\usepackage{amsmath,amssymb,amsthm}
\usepackage{algorithm}
\usepackage[noend]{algpseudocode}
\usepackage{booktabs}
\usepackage{tabularx}
\usepackage{graphicx}
\usepackage{tikz}
\usepackage{placeins}
\usepackage{enumitem}
\usepackage{xcolor}
\usepackage{xurl}
\usepackage{seqsplit}
\usepackage[numbers,sort&compress]{natbib}
\usepackage{hyperref}
\usetikzlibrary{arrows.meta,positioning,fit,backgrounds}

\definecolor{RCBlue}{HTML}{1769AA}
\definecolor{RCCharcoal}{HTML}{262B33}
\definecolor{RCPaper}{HTML}{F7F8FA}
\hypersetup{
  colorlinks=true,
  linkcolor=RCBlue,
  citecolor=RCBlue,
  urlcolor=RCBlue,
  pdftitle={Ready Cohorts: Bounding GPU Opportunity and Avoiding Host Round Trips in LLM-Agent Control},
  pdfauthor={Josef Chen},
  pdfsubject={A theory and measurement study of GPU execution for deterministic LLM-agent control},
  pdfkeywords={LLM agents, GPU systems, batching, deadline scheduling, CUDA graphs}
}
\setlist{leftmargin=*,nosep}
\graphicspath{{figures/}}

\newtheorem{proposition}{Proposition}
\newtheorem{lemma}{Lemma}

\newcommand{\pstar}{P^\star}
\newcommand{\claimid}[1]{\ignorespaces}
\newcommand{\sha}[1]{\texttt{\seqsplit{#1}}}
\input{generated/results-macros.tex}

\title{\vspace{-1.4em}\textbf{Ready Cohorts:}\\
Bounding GPU Opportunity and Avoiding Host Round Trips in LLM-Agent Control}
\author{Josef Chen\\
\small Independent Researcher}
\date{August 2026}

\begin{document}
\raggedbottom
\twocolumn[
  \begin{@twocolumnfalse}
  \maketitle
  \vspace{-1.25em}
  \begin{abstract}
    \input{sections/00_abstract}
  \end{abstract}
  \vspace{0.45em}
  \end{@twocolumnfalse}
]

\input{sections/01_introduction}
\input{sections/02_scope}
\input{sections/03_model}
\input{sections/04_exact_packing}
\input{sections/05_trace_method}
\input{sections/06_trace_results}
\input{sections/07_resident_method}
\input{sections/08_resident_results}
\input{sections/09_joint_interpretation}
\input{sections/10_related_work}
\input{sections/11_limitations}
\input{sections/12_reproducibility}
\input{sections/13_conclusion}

\bibliographystyle{plainnat}
\bibliography{references}

\clearpage
\onecolumn
\appendix
\input{sections/appendix}

\end{document}

%% file: generated/results-macros.tex
\newcommand{\TraceSessions}{851}
\newcommand{\TraceSpans}{9{,}031}
\newcommand{\TracePrimaryEvents}{651{,}123}
\newcommand{\TraceFixedPct}{30.19}
\newcommand{\TraceExactPct}{43.00}
\newcommand{\TraceUpperPct}{45.85}
\newcommand{\TraceGapClosurePct}{81.83}
\newcommand{\TraceFixedToExactPP}{12.81}
\newcommand{\TraceExactToUpperPP}{2.85}

\newcommand{\TracePrimaryExactMinPct}{42.43}
\newcommand{\TracePrimaryExactMaxPct}{43.41}
\newcommand{\TracePAtTenKThirtyTwoPct}{22.2}
\newcommand{\TracePAtTenKSixtyFourPct}{0.0}
\newcommand{\TracePAtHundredKThirtyTwoPct}{66.8}
\newcommand{\TracePAtHundredKSixtyFourPct}{66.0}
\newcommand{\TracePAtHundredKOneTwentyEightPct}{48.4}
\newcommand{\TracePAtHundredKTwoFiftySixPct}{43.0}
\newcommand{\TraceReplayCells}{180}
\newcommand{\TraceReplayRepetitions}{540}

\newcommand{\ResidentCells}{36}
\newcommand{\ResidentRows}{3{,}240}

\newcommand{\ResidentLegalInvocations}{14{,}557{,}440}
\newcommand{\ResidentRatioMin}{1.19}
\newcommand{\ResidentRatioMax}{2.39}
\newcommand{\ResidentPrimaryUsMin}{258}
\newcommand{\ResidentPrimaryUsMax}{309}
\newcommand{\HostPrimaryUsMin}{467}
\newcommand{\HostPrimaryUsMax}{625}
\newcommand{\PrimarySavedUsMin}{194}
\newcommand{\PrimarySavedUsMax}{363}
\newcommand{\PrimaryFloorRatioMin}{6.60}
\newcommand{\PrimaryFloorRatioMax}{8.17}
\newcommand{\NativePlacements}{5}
\newcommand{\NativeCells}{60}
\newcommand{\NativeRows}{12{,}000}
\newcommand{\NativeRatioMin}{1.07}
\newcommand{\NativeRatioMax}{1.99}

%% file: sections/00_abstract.tex
\claimid{RC-013}LLM-agent services repeatedly execute small deterministic
transitions between model and tool calls: route an outcome, update state, and
emit the next effect. We ask when this control path exposes enough concurrent
work for GPU execution, and what changes when a GPU-computed route decision
remains on device. We formalize the \emph{ready-cohort boundary} using
fixed-partition share $F$, exact offline share $\pstar$, local upper bound $U$,
and online achieved share $A$.

\claimid{RC-001}\claimid{RC-003}\claimid{RC-004}For zero service time,
unlimited capacity, and equal relative launch deadlines, a specialized dynamic
program computes $\pstar$ exactly. \claimid{RC-005}\claimid{RC-006}In a
prospectively frozen stationary Poisson replay of one pinned
\TraceSessions-session public trace panel, the primary condition at 100{,}000
target active sessions, $K=256$, and a 50 ms launch deadline gives
$F=\TraceFixedPct\%$, $\pstar=\TraceExactPct\%$, and $U=\TraceUpperPct\%$.
Exact packing recovers \TraceGapClosurePct\% of the opportunity lost at fixed
window boundaries. \claimid{RC-015}The outcome-derived route key is a
conditioning proxy, not proof of executable identity.

\claimid{RC-008}\claimid{RC-009}A separate mechanism study keeps a
GPU-computed binary decision on device instead of returning four bytes to the
host and redispatching. Across four named GPU placements, the device-resident
path is faster in all \ResidentCells{} configurations; within-placement
row-median ratios range from \ResidentRatioMin$\times$ to
\ResidentRatioMax$\times$. \claimid{RC-010}Across both admissible mechanisms,
all \ResidentLegalInvocations{} tested batched invocations match a separately
implemented host oracle. \claimid{RC-011}A fixed nested device graph that
removes no host decision is slower in all \NativeCells{} configurations across
five placements.

Together, the studies establish two measurable gates for GPU agent control:
deadline-feasible cohort supply and observation placement. They expose
schedulable work missed by fixed windows. A joined finite online runtime is
required to measure $A$, CPU displacement, and service-level benefit.

%% file: sections/01_introduction.tex
\section{Introduction}
\label{sec:introduction}

An agent runtime performs deterministic control work between model and tool
calls. It parses a typed outcome, advances a state machine, checks policy and
budget state, selects a route, and emits the next effect. Each transition is
small beside inference, but a service with many concurrent sessions executes
them continuously.

This control path is already a datacenter concern. A recent production and
open-source characterization found that agentic workflows repeatedly cross the
CPU--GPU boundary, place host orchestration on the critical path, and create
bursty CPU demand \cite{yang2026agenticarchitecture}. That result establishes
the pressure. It leaves a more specific systems question open: when can the
deterministic transition itself be grouped and executed profitably on a GPU?

Similarity across agents is not enough. A GPU route needs a cohort above its
hardware and runtime crossover, all members must share executable semantics,
and they must become ready before their launch deadlines. A decision returned
to the host after every step also incurs copy, synchronization, branch, and
redispatch overhead. We call the interface between cohort supply and
observation placement the \emph{ready-cohort boundary}.

For each route and hardware/runtime configuration, let $K$ begin a measured
safe suffix: every tested cohort from $K$ through a declared maximum clears a
chosen admissible baseline. A workload supplies events with release times,
launch deadlines, and declared grouping keys. The central workload question is
how much same-group work can be packed above $K$. The corresponding mechanism
question is how much of the control chain can proceed without exposing an
intermediate decision to the host.

We answer these questions in two experiments. A public agent-trace replay
compares fixed-window eligibility with an exact offline packing optimum and a
local upper bound. A separate CUDA study holds the state transition and route
bodies fixed while changing where one GPU-computed decision is observed. An
earlier device-launch design provides a negative control. The experiments are
kept numerically separate because the trace threshold $K=256$ is swept rather
than measured for the resident-policy source, and the 32-epoch mechanism
horizon is not inferred from the trace.

The paper makes four contributions:

\begin{enumerate}[label=\arabic*.]
  \item \textbf{A ready-cohort formulation.} We separate the hardware
  threshold $K$, fixed-partition share $F$, exact offline share $\pstar$, local
  bound $U$, and online achieved share $A$. Each quantity has one source and a
  defined inference boundary.
  \item \textbf{An exact opportunity instrument.} Under zero service,
  unlimited capacity, and equal relative launch deadlines, a specialized
  dynamic program computes $\pstar$. This gives a reproducible workload
  measurement without claiming algorithmic priority.
  \item \textbf{Trace evidence for hidden cohort supply.} In the frozen
  primary replay, exact packing raises eligible share from
  \TraceFixedPct\% to \TraceExactPct\% and recovers
  \TraceGapClosurePct\% of the fixed-window alignment gap. The full grid also
  identifies regimes with no qualifying cohort.
  \item \textbf{Cross-provider mechanism evidence.} Across four named GPU
  placements and \ResidentCells{} cells, retaining the tested decision on
  device is \ResidentRatioMin$\times$ to \ResidentRatioMax$\times$ faster
  than the matched host-mediated GPU path. A fixed nested graph loses in all
  \NativeCells{} cells across five placements, ruling out device launch alone
  as the explanation.
\end{enumerate}

These experiments resolve two prerequisite questions for a GPU control plane:
whether enough work can coexist under a declared grouping, and whether one
device-side decision avoids measurable host-mediated overhead. The next system
must join them with finite service and CPU fallback, then measure $A$, CPU use,
raw tail latency, exact effects, task utility, and inference interference.

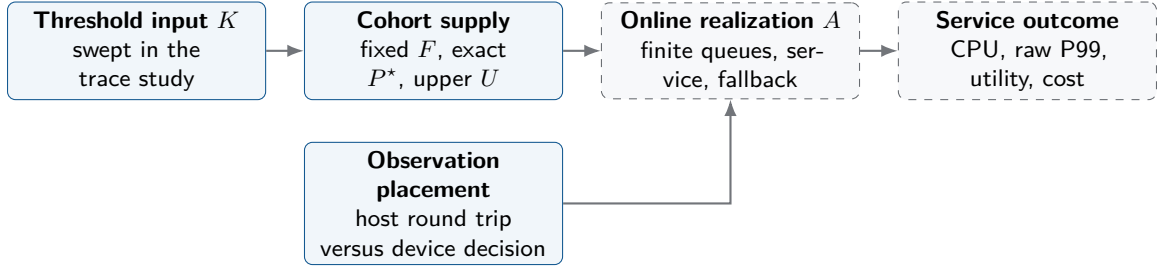
\begin{figure*}[t]
  \centering
  \begin{tikzpicture}[
    node distance=0.75cm and 0.50cm,
    font=\sffamily\small,
    evidence/.style={draw=RCBlue!85!black,rounded corners=3pt,fill=RCBlue!6,
      text width=3.18cm,minimum height=1.12cm,align=center},
    next/.style={draw=RCCharcoal!70,rounded corners=3pt,fill=RCPaper,
      text width=3.18cm,minimum height=1.12cm,align=center,dashed},
    arrow/.style={-{Latex[length=2.2mm]},thick,draw=RCCharcoal!65}
  ]
    \node[evidence] (k) {\textbf{Threshold input $K$}\\swept in the trace study};
    \node[evidence,right=of k] (opp) {\textbf{Cohort supply}\\fixed $F$, exact $\pstar$, upper $U$};
    \node[next,right=of opp] (a) {\textbf{Online realization $A$}\\finite queues, service, fallback};
    \node[next,right=of a] (benefit) {\textbf{Service outcome}\\CPU, raw P99, utility, cost};
    \draw[arrow] (k) -- (opp);
    \draw[arrow] (opp) -- (a);
    \draw[arrow] (a) -- (benefit);
    \node[evidence,below=0.55cm of opp] (res)
      {\textbf{Observation placement}\\host round trip versus device decision};
    \draw[arrow] (res) -| (a);
  \end{tikzpicture}
  \caption{Evidence map. Solid boxes are the controlled or measured components
  in this paper; dashed boxes define the joined online experiment and its
  service-level outcomes. The shared notation connects the studies without
  treating the swept trace threshold as the resident source's measured
  crossover.}
  \label{fig:claim-boundary}
\end{figure*}

%% file: sections/02_scope.tex
\section{Execution and semantic scope}
\label{sec:scope}

\subsection{Unit of execution}

Our unit is a deterministic control transition inside an agent runtime. It is
smaller than an LLM request, a tool, or a complete workflow. A transition reads
typed state and an already available event, then produces new typed state, a
route, and possibly an effect descriptor. Model inference, tool execution, and
the privileged commit of external effects remain outside the unit.

This separation matters for both performance and safety. A GPU may compute
that an agent should call a tool, but the current paper does not ask the GPU to
hold cloud credentials, create a virtual machine, or commit a side effect. A
future runtime can keep pure state evolution on device while a CPU or DPU
validates and commits ordered effect descriptors.

\subsection{Compatibility and exactness}

Two events are compatible only when one implementation can process them
together without changing the declared semantics. The trace study cannot
verify that condition. Its narrowest grouping is an outcome-derived route-key
proxy, and even that key omits state-machine node, schema, arguments, policy
context, and the identities in multi-tool outcomes. Coarser event classes and a
pooled route are diagnostics. None of these groupings proves that fusion is
semantically admissible.

Correctness in the mechanism studies means exact equality of every declared
state field and the full decision sequence against a separately implemented
host oracle. It does not establish correct recovery, distributed ordering, or
real tool effects. Checksums are provenance guards rather than the definition
of equality.

\subsection{Launch deadlines, not completion SLOs}

Every trace event has a latest admissible launch time. This is a bound on how
long the scheduler may wait to form a cohort. Kernel service and queueing after
launch are zero in the mathematical model. The 50 ms primary value is therefore
not a 50 ms completion service-level objective. Finite service enters only in
the proposed online system.

\begin{table}[t]
  \centering
  \footnotesize
  \begin{tabularx}{\columnwidth}{@{}lX@{}}
    \toprule
    In scope & Out of scope in this paper \\
    \midrule
    Typed deterministic transition & Arbitrary Python on a GPU \\
    Route-key-conditioned cohorts & Verified semantic fusion \\
    Latest launch time & End-to-end completion SLO \\
    Device-resident pure decision & Privileged tool or VM lifecycle \\
    Exact tested trajectory & Distributed commit and recovery \\
    Named GPU placements & Hardware-population inference \\
    \bottomrule
  \end{tabularx}
  \caption{The execution contract. These exclusions are claim boundaries, not
  assumed properties of a future system.}
  \label{tab:scope}
\end{table}

%% file: sections/03_model.tex
\section{The ready-cohort model}
\label{sec:model}

\subsection{Events and the hardware threshold}

Let $E$ be a set of control events. Event $i$ has release time $t_i$, latest
admissible launch time $d_i\geq t_i$, and route $r_i$. A route in the
mathematical model means that one implementation can process the grouped events
without changing their declared transition semantics. A trace label does not
establish that condition by itself.

For route $r$ and hardware/runtime configuration $h$, define
$K_r(h,v,H)$ as the start of a measured safe suffix: every tested cohort size
from $K_r$ through a declared maximum beats a named admissible baseline. The
model extrapolates that monotone threshold and has no upper batch limit. The
arguments record state visibility $v$ and observation-free horizon $H$ because
transfer, synchronization, and temporal fusion can move the crossover. We
write $K_r$ when the configuration is fixed.

A batch $(\tau,B)$ is feasible when

\begin{equation}
  |B|\geq K_r,\quad r_i=r\ \forall i\in B,\quad
  t_i\leq\tau\leq d_i\ \forall i\in B.
  \label{eq:feasible-batch}
\end{equation}

A schedule contains feasible batches and assigns each event at most once. The
model allows batches to have more than $K_r$ events. It has no maximum batch
size, service time, or device-capacity constraint.

\subsection{Four shares with different meanings}

Let $\pi$ be a frozen partition of time into non-overlapping windows. Each
event belongs to a bucket defined by its window and route. If bucket $b$ has
$n_b$ events, fixed-partition eligibility is

\begin{equation}
  F(\pi,K)=\frac{1}{|E|}\sum_b n_b\,\mathbf{1}[n_b\geq K_{r(b)}].
  \label{eq:fixed}
\end{equation}

$F$ is exact only for schedulers constrained to $\pi$. Sliding deadlines can
join events across a fixed boundary, so $F$ is not a universal ceiling.

Let $\pstar$ be the maximum number of assigned events in any feasible offline
schedule, divided by $|E|$. It uses future knowledge and therefore measures
opportunity rather than online performance.

For route $r$, define the active count at time $\tau$ as

\begin{equation}
  Q_r(\tau)=|\{j:r_j=r,\ t_j\leq\tau\leq d_j\}|.
\end{equation}

Event $i$ is locally eligible if $Q_{r_i}(\tau)\geq K_{r_i}$ at some
$\tau\in[t_i,d_i]$. The local-overlap share is

\begin{equation}
  U=\frac{1}{|E|}\sum_i \mathbf{1}\!\left[
  \max_{\tau\in[t_i,d_i]}Q_{r_i}(\tau)\geq K_{r_i}\right].
  \label{eq:upper}
\end{equation}

$U$ may count incompatible opportunities that reuse the same events. It is an
upper bound, not necessarily an achievable schedule.

Finally, $A$ is the share accelerated within deadline by a specified finite
online runtime, including its queue, service, fallback, and capacity rules. To
compare $A$ with $\pstar$, each accelerated batch must obey the same route
grouping, threshold $K_r$, and launch deadline as the offline model; finite
service and capacity may impose stricter constraints. Its denominator is the
identical event set $E$ and retained horizon used by $\pstar$; late, missing,
failed, and fallback events remain in that denominator. The current paper does
not measure $A$.

\begin{proposition}[Boundary ordering]
\label{prop:ordering}
If every batch admitted by the frozen partition is admissible under the event
deadlines, and the batches counted by $A$ obey the same route, threshold,
deadline, and event-set constraints as the offline model, then
\begin{equation}
  F\leq\pstar\leq U,\qquad A\leq\pstar.
  \label{eq:ordering}
\end{equation}
\end{proposition}

\begin{proof}
Every bucket with $n_b\geq K_{r(b)}$ is a feasible batch, and buckets are
disjoint, so the fixed schedule attains $F$. Every event in any feasible batch
overlaps at least $K_{r_i}$ same-route intervals at that batch's launch time,
so it is counted by $U$. An online schedule is one member of the offline
feasible set.
\end{proof}

\subsection{Compatibility has a measurable tax}

Suppose grouping $g_f$ refines $g_c$, so every fine bucket lies inside one
coarse bucket. For fixed $\pi$ and $K$,

\begin{equation}
  F(\pi,K,g_f)\leq F(\pi,K,g_c).
  \label{eq:grouping}
\end{equation}

This inequality measures fragmentation under a declared grouping. It does not
make coarse fusion admissible.
A pooled kernel may change control flow, memory access, actions, or numerical
trajectory. The route-key proxy is therefore the narrowest grouping in the
trace study.

\begin{table}[t]
  \centering
  \footnotesize
  \begin{tabularx}{\columnwidth}{@{}llX@{}}
    \toprule
    Symbol & Source & Meaning \\
    \midrule
    $K_r$ & hardware & route-specific profitable cohort \\
    $F$ & workload, policy & fixed-partition eligible share \\
    $\pstar$ & workload, offline & exact schedulable share \\
    $U$ & workload, local & necessary-overlap upper bound \\
    $A$ & online runtime & achieved accelerated share \\
    \bottomrule
  \end{tabularx}
  \caption{The quantities cannot be substituted for one another. In
  particular, $K=256$ in the trace sweep is not a measured $K_r$ for the
  resident-policy mechanism.}
  \label{tab:notation}
\end{table}

%% file: sections/04_exact_packing.tex
\section{Exact packing under equal relative deadlines}
\label{sec:exact}

The trace experiment uses $d_i=t_i+\delta$ for a common $\delta$. This special
case permits an exact dynamic program. We keep the assumptions visible because
the result does not extend unchanged to arbitrary deadlines or finite service.

\begin{proposition}[Deadline launch times]
\label{prop:deadlines}
Under zero service time and unlimited simultaneous capacity, some optimal
offline schedule launches every batch at an event deadline.
\end{proposition}

\begin{proof}
Take a nonempty feasible batch launched at $\tau$ and let $d_{\min}$ be the
smallest deadline among its events. Feasibility gives $\tau\leq d_{\min}$.
Every assigned event has release at most $\tau$ and deadline at least
$d_{\min}$, so moving the launch to $d_{\min}$ preserves every assignment.
Same-route batches moved to the same time can be merged; different routes may
co-launch under the unlimited-capacity assumption.
\end{proof}

\begin{lemma}[Contiguous block form]
\label{lem:blocks}
For one route with equal relative deadlines, there is an optimum whose batches
are disjoint contiguous blocks in sorted release order.
\end{lemma}

\begin{proof}
Sort releases so $t_1\leq\cdots\leq t_n$, and order batches by launch time.
Suppose $i<j$, event $i$ is assigned to a later batch, and event $j$ to an
earlier batch. Their feasibility gives
\[
 t_i\leq t_j\leq\tau_{\mathrm{early}}\leq
 \tau_{\mathrm{late}}\leq t_i+\delta\leq t_j+\delta.
\]
Exchanging the two assignments is therefore feasible. Repeated exchanges
remove crossings. An unassigned event between the first and last release of a
batch can then be added to that batch: its release is no later than the batch
launch, and its equal-length deadline is no earlier than the first event's
deadline. Thus each batch may be taken as a contiguous block.
\end{proof}

Let $D[j]$ be the maximum number of assigned events among the first $j$
releases of one route, with $D[0]=0$. A feasible block ending at $j$ starts at
some $i\leq j-K_r+1$ and satisfies $t_j-t_i\leq\delta$. Hence

\begin{equation}
\begin{aligned}
D[j]=\max\biggl\{&D[j-1],\\[-2pt]
 &\max_{\substack{i\leq j-K_r+1\\t_j-t_i\leq\delta}}
 \bigl(D[i-1]+j-i+1\bigr)\biggr\}.
\end{aligned}
\label{eq:dp}
\end{equation}

The inner term is $j+\max_i(D[i-1]-i+1)$ over a sliding interval of valid
starts. A grouped sort, route slices, a moving left pointer, and a monotone
deque give an $O(N\log N)$ evaluator including sorting. The frozen evaluator
used for the reported results instead scans the full group array once per route
and finds each left boundary by binary search. Its bound is
$O(NR+\sum_r n_r\log n_r)$ for $R$ routes and $n_r$ events per route, which is
quadratic in the worst case. This affects solver cost, not the returned optimum.
Back-pointers recover a maximum-cardinality witness.

\begin{algorithm}[t]
  \caption{Equal-relative-deadline packing for one route}
  \label{alg:packing}
  \begin{algorithmic}[1]
    \Require sorted integer releases $t_1,\ldots,t_n$, threshold $K$, deadline $\delta$
    \State $D[0]\gets 0$; initialize empty monotone deque $Q$
    \For{$j\gets1$ to $n$}
      \State $D[j]\gets D[j-1]$
      \State $i\gets j-K+1$
      \If{$i\geq1$}
        \State insert $(i,D[i-1]-i+1)$ into $Q$, removing smaller tails
      \EndIf
      \State $\ell\gets\Call{LowerBound}{t,t_j-\delta}$
      \State remove heads whose start index is smaller than $\ell$
      \If{$Q$ is nonempty}
        \State $D[j]\gets\max(D[j],j+Q.\mathrm{head.value})$
      \EndIf
    \EndFor
    \State \Return $D[n]$
  \end{algorithmic}
\end{algorithm}

\paragraph{Exact clock.}
The implementation rounds releases and relative deadlines to integer
nanoseconds and uses inclusive comparisons. The claim is exact on that clock.
It does not rely on a floating tolerance. The implementation agrees with
subset brute force on tiny instances, including route-specific thresholds and
adversarial boundary cases.

\paragraph{Role in this paper.}
For fixed $K$ and $\delta$, the trace oracle is equivalent to selecting the
maximum number of ordered points that can be partitioned into clusters of at
least $K$ points and diameter at most $\delta$. This is the fixed-radius
decision form of one-dimensional $r$-gathering with outliers: related work
often fixes an outlier budget and minimizes radius, whereas we fix the radius
and maximize retained points \cite{aggarwal2010anonymity,akagi2015rgather,
nakano2019rgather}. Classical scheduling also studies compatible batching
under release times and deadlines \cite{barnoy2009throughput}. We use
Equation~\eqref{eq:dp} only as an exact trace oracle and make no algorithmic-priority
claim.

%% file: sections/05_trace_method.tex
\section{Trace study}
\label{sec:trace-method}

\subsection{Pinned public source}

The source is the complete \texttt{tau2\_airline}, \texttt{tau2\_retail}, and
\texttt{tau2\_telecom} subset of the public Exgentic agent-trace dataset. The
dataset revision is
\sha{70036b93a04e61b0ea2706a68b962f4f26774587}
and the Parquet conversion revision is
\sha{f7c94012d0bfbf66fe4d6ed627699508bbb555ff}; all 19 retained shard hashes
match commit-resolved URLs \cite{exgentic2026traces}. The domain labels correspond to the airline, retail,
and telecom environments released with $\tau^2$-Bench
\cite{barres2026tau2bench}, which builds on the original $\tau$-bench framework
\cite{yao2024taubench}. The selected panel contains
\TraceSessions{} sessions and \TraceSpans{} recorded LLM spans across four
harnesses.

\claimid{RC-015}Each span completion becomes one candidate control event. Its route key is
derived from the recorded outcome: final or text, error, one named tool, or a
generic multi-tool outcome. It omits benchmark and harness semantics,
state-machine node, schema and version, arguments, policy context, and the tool
identities inside a multi-tool outcome. The extraction retains timestamps,
public identifiers, counts, lengths, and route labels. It excludes prompt text,
tool arguments, and tool results. Source Parquet files and both derived outputs
are content-hashed.

The selected data contain 70 route labels. Eight labels occur in more than one
benchmark domain, and 28 occur under more than one harness. The generic
\texttt{tool:<multi>} label alone covers 325 spans and 122 distinct recorded
tool-name encodings. The source audit also retains 52 failed-status spans, 52
nonpositive-duration spans, and 617 overlapping span starts. We do not remove
these observations after inspecting outcomes.

This construction treats a model completion as the point at which a control
transition could become ready. The route key is a declared conditioning proxy,
not proof of executable compatibility. The construction does not claim that
the recorded harnesses implemented our transition or that their span durations
equal model service times.

\subsection{Stationary replay}

We use the prospectively frozen stationary swarm model from trace replay 003. Session
arrivals follow a homogeneous Poisson process. A session template is sampled
uniformly from the fixed empirical panel. For target mean active population
$C$, the arrival rate is $C$ divided by mean template duration. Arrivals begin
one maximum template duration before measurement so that the retained
60-second interval is in steady state under the model. Fixed partitions are
origin-aligned half-open windows of width $\delta$; $F$ is conditional on that
phase, whereas $\pstar$ is not.

The exact-packing grid contains:

\begin{itemize}
  \item $C\in\{1{,}000,10{,}000,100{,}000\}$;
  \item $\delta\in\{10,25,50,100,250\}$ ms;
  \item pooled, event-class, and route-key grouping;
  \item $K\in\{32,64,128,256\}$; and
  \item three replay seeds from root seed 20260811.
\end{itemize}

\claimid{RC-007}Every generated event is retained. The three population values crossed with
three seeds produce nine generated swarms. Reusing each swarm across deadlines,
groupings, and thresholds yields \TraceReplayCells{} design cells and
\TraceReplayRepetitions{} cell-seed rows. For each row we compute $F$, $\pstar$,
$U$, the exact batch count, and

\begin{equation}
  G=\frac{\pstar-F}{U-F}\quad\text{when }U>F.
  \label{eq:gap-closure}
\end{equation}

$G$ describes how much of the local boundary-alignment opportunity is jointly
packable. The reported primary $G$ is the mean of the three per-seed ratios,
not the ratio formed from three mean shares.

\subsection{Prospective freeze and validity gates}

The exact solver, integer-clock contract, inputs, grid, primary cell, and
directional hypothesis were frozen before any exact trace outcome was
computed according to the artifact-recorded chronology. The plan was not
externally registered or independently timestamped. Earlier fixed-window and local-bound results were permitted design
data. Five invariants gate interpretation: $F\leq\pstar\leq U$; monotonicity
under grouping coarsening; monotonicity with $\delta$; antitonicity with
$K$; and equality of all three shares whenever the previously computed lower
and upper bounds coincide.

The primary cell was $C=100{,}000$, route-key grouping, $K=256$, and
$\delta=50$ ms.
The directional pilot hypothesis was only $\pstar>F$. The three seeds quantify
Monte Carlo variation under one fixed panel and arrival model. They are not
independent traces from a deployment population, so no population $p$-value or
confidence interval is reported.

\begin{table}[t]
  \centering
  \footnotesize
  \begin{tabularx}{\columnwidth}{@{}lX@{}}
    \toprule
    Frozen item & Value \\
    \midrule
    Source & Exgentic tau2 panel, two commit hashes \\
    Empirical unit & one of \TraceSessions{} session templates \\
    Event & recorded LLM-span completion \\
    Grouping & outcome-derived route-key proxy \\
    Replay & stationary Poisson session arrivals \\
    Clock & integer nanoseconds, inclusive endpoints \\
    Capacity/service & unlimited / zero \\
    Uncertainty & three seeds conditional on one panel \\
    \bottomrule
  \end{tabularx}
  \caption{Trace-study contract. The replay is a controlled load model rather
  than a model of production arrival statistics.}
  \label{tab:trace-contract}
\end{table}

%% file: sections/06_trace_results.tex
\section{Trace-conditioned opportunity}
\label{sec:trace-results}

\claimid{RC-007}All \TraceReplayRepetitions{} cell-seed rows from nine generated swarms pass
the five validity gates. In
the prospectively frozen primary cell, the replay generates a mean of
\TracePrimaryEvents{} events. Fixed windows admit \TraceFixedPct\%, while exact
sliding-deadline packing admits \TraceExactPct\%. The exact share ranges from
\TracePrimaryExactMinPct\% to \TracePrimaryExactMaxPct\% across the three
seeds. The local bound is \TraceUpperPct\%.

\begin{table}[t]
  \centering
  \footnotesize
  \input{generated/trace-primary-table.tex}
  \caption{Prospectively frozen primary trace cell: 100{,}000 target active sessions,
  route-key grouping, $K=256$, and a 50 ms launch deadline. Shares and counts are
  means over three replay seeds conditional on one fixed panel.}
  \label{tab:trace-primary}
\end{table}

Exact packing therefore gains \TraceFixedToExactPP{} percentage points over the
frozen partition and remains \TraceExactToUpperPP{} points below the local
upper bound. The mean of the per-seed gap closures is
\TraceGapClosurePct\%. This distinction matters operationally: the fixed
partition leaves model-admissible cross-boundary cohorts unused, while $U$ still
overcounts opportunities that cannot all be selected together.

\begin{figure*}[t]
  \centering
  \includegraphics[width=0.82\linewidth]{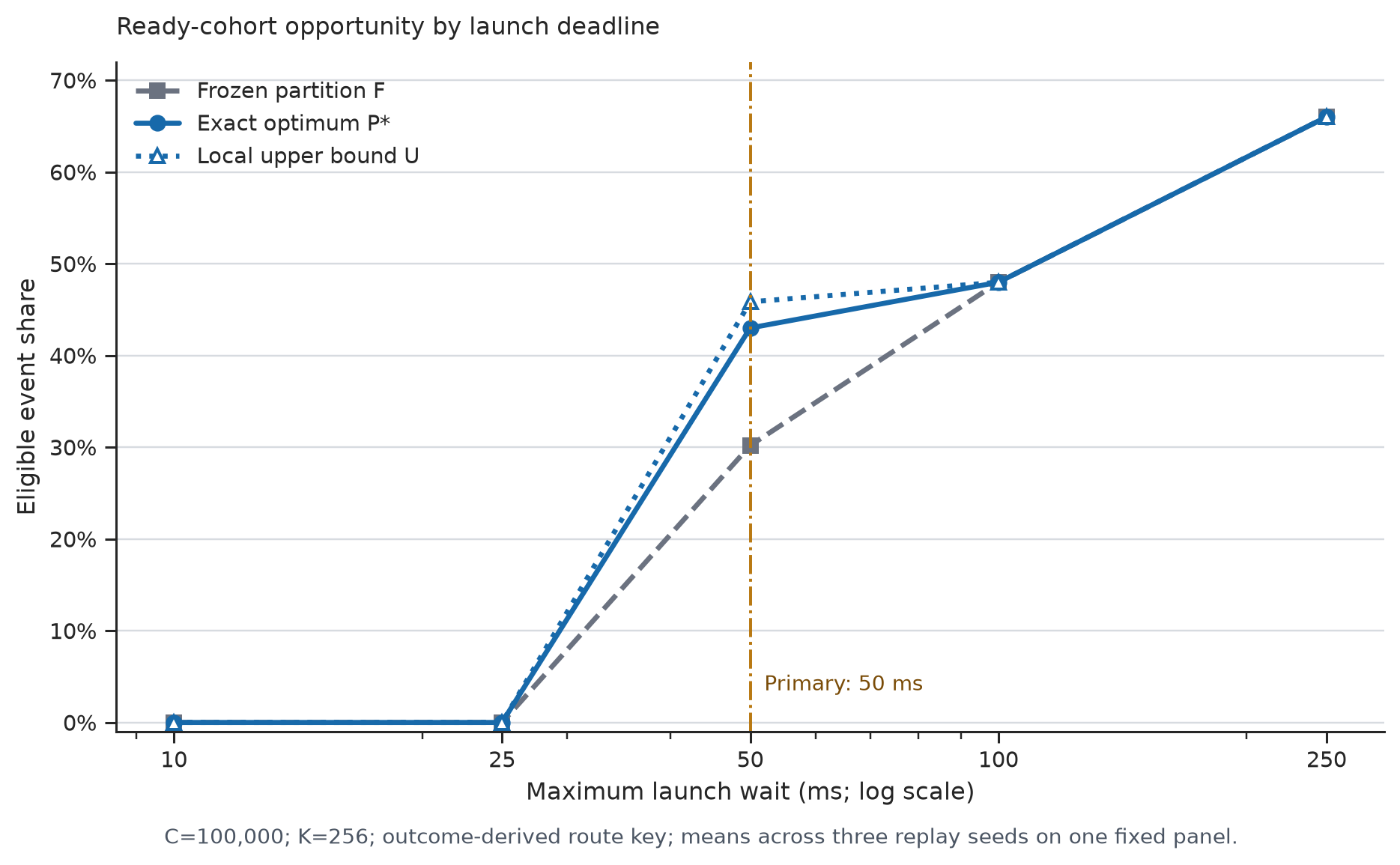}
  \caption{Route-key-conditioned opportunity at $C=100{,}000$ and $K=256$. The
  prospectively frozen primary deadline is 50 ms. Values are means across three replay
  seeds conditional on one panel and one arrival model, not population
  confidence intervals.}
  \label{fig:trace-frontier}
\end{figure*}

\subsection{Cohort supply collapses below the boundary}

Under route-key grouping at $K=256$, $\pstar$ is zero for every tested deadline when
$C\leq10{,}000$. Even at $C=100{,}000$, it remains zero at 10 and 25 ms. Thus
a large nominal swarm does not imply that a route sees a profitable cohort
inside a short launch budget.

Reducing $K$ changes the boundary sharply. At $C=10{,}000$ and 50 ms,
$\pstar$ is \TracePAtTenKThirtyTwoPct\% for $K=32$ but
\TracePAtTenKSixtyFourPct\% for $K=64$. At $C=100{,}000$ and 50 ms, the exact
shares for $K=32,64,128,256$ are respectively
\TracePAtHundredKThirtyTwoPct\%, \TracePAtHundredKSixtyFourPct\%,
\TracePAtHundredKOneTwentyEightPct\%, and
\TracePAtHundredKTwoFiftySixPct\%. Lowering the hardware threshold can matter
more than admitting another boundary-crossing cohort.

The surface gives an online runtime a measurable workload budget. Under the
model, $F$ is the fixed-window baseline, $\pstar$ is the most that any legal
scheduler can attain, and $\pstar-F$ is the headroom available to a better
packing policy. A real implementation must measure its route-specific $K$ and
report how much of that headroom becomes achieved share $A$.

%% file: generated/trace-primary-table.tex
\begin{tabular}{@{}lr@{}}
  \toprule
  Quantity & Primary value \\
  \midrule
  Fixed-partition share $F$ & 30.19\% \\
  Exact offline share $P^\star$ & 43.00\% \\
  Local-overlap bound $U$ & 45.85\% \\
  Alignment-gap closure & 81.83\% \\
  Mean generated events & 651{,}123 \\
  Mean exact batches & 1{,}046.3 \\
  \bottomrule
\end{tabular}

%% file: sections/07_resident_method.tex
\section{Device-resident decision study}
\label{sec:resident-method}

The trace study measures cohort supply. This study isolates the second gate:
the placement of a decision over state that is already resident on the GPU. It
is a mechanism experiment rather than an implementation of trace-driven route
compaction.

\subsection{Matched mechanisms}

The frozen CUDA program maintains a 16-byte synthetic state for each agent and
two route bodies, represented by per-epoch pre-instantiated branch or path graph
executables. At each epoch, a GPU predicate computes one global binary route
decision. All mechanisms use the same initialized state, predicate, route
functions, block size, and implied route sequence.

\paragraph{Host round trip.}
Launch the predicate graph, copy its four-byte result to pinned host memory,
synchronize, and launch the selected per-epoch route graph from the host.

\paragraph{Device resident.}
Launch one root graph. A one-thread selector reads the predicate on device and
tail-launches one of the uploaded per-epoch path graphs. Each nonfinal path
executes the selected body and the next predicate and selector, so execution
continues for $H$ epochs without exposing the decision to the host.

\paragraph{No-decision floor.}
Replay the oracle route sequence as one graph while omitting predicate and
selection work. This mechanism is a structural floor, not an admissible online
scheduler.

Graph creation, instantiation, and upload are outside steady-state timing for
all mechanisms. State reset, result copy, and validation are also outside. The
predicate copy and synchronization remain inside the host-roundtrip interval.
The primary metric is cohort-horizon wall time per batched invocation, averaged
within each recorded technical row.

\subsection{Grid and independent unit}

The full grid is $N\in\{256,2048,16384\}$ agents and
$H\in\{2,8,32\}$ epochs. Each mechanism-cell has five warmups, three
calibration samples, and 30 measured rows. One common batch count, calibrated
from the fastest mechanism, is applied to all mechanisms in a cell; each row
runs for at least 100 ms of aggregate timed work, subject to a frozen cap.

The four named placements are a local GTX 1660 Ti, a Modal L4, a RunPod L4,
and a Lambda H100 SXM5. The CUDA source hash is
\sha{4b5cdcb9496a734bd7801d5c419efb8eceb72fd6962800520101e89676d204da}.
Provider receipts, device UUIDs, images, binaries, raw files, and manifests are
bound where provider interfaces permit.

A named placement is the outer performance unit. Thirty rows reduce technical
timing noise inside one placement and are not independent hardware
replications. Effect magnitudes and directions are therefore reported by
placement, with L4 $n=2$, H100 $n=1$, and GTX $n=1$.

\subsection{Correctness contract}

A separately written host implementation computes both route functions and
the predicate without calling the device transition functions. After every
batched invocation, the experiment compares all four fields of every final
agent state. It also compares the complete epoch-by-epoch decision trace for
the host and resident mechanisms. Any mismatch, illegal launch, setup error,
OOM, crash, or nonpositive timing blocks performance interpretation.

The no-decision floor receives the oracle route sequence by construction. Its
final state is checked independently, but its decision sequence is not an
independent observation. This is why the paper reports admissible-mechanism
correctness separately from total tested work.

\subsection{Prospective-plan deviations}

A bounded local engineering smoke preceded the freeze. The full local run
used the same physical GPU shortly afterward. The prospectively frozen bounded cloud
stage allowed Modal plus one of RunPod or Lambda; the final report retains both
external providers as a descriptive scope expansion, one placement beyond the
frozen first-stage count. Neither fact strengthens the sampling claim.

The primary directional cell was $N=256,H=32$. Horizon scaling is reported as
an exploratory diagnostic because ``advantage increases'' was not given a
complete operational definition before the run.

\begin{figure*}[t]
  \centering
  \includegraphics[width=0.9\linewidth]{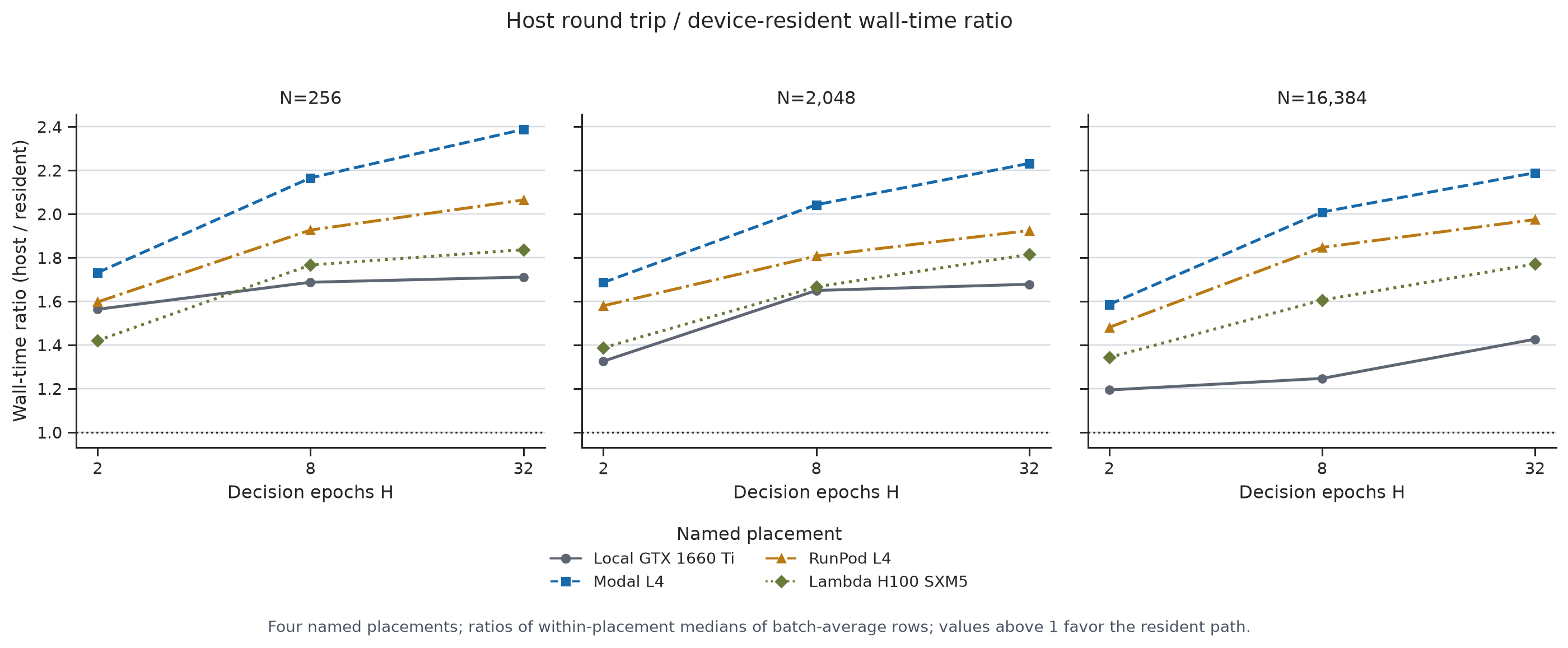}
  \caption{Within-placement ratio of batch-average row medians for the host
  round trip over the device-resident mechanism. Ratios above one favor the
  resident path. Placement is the outer unit; the plotted rows are technical
  repetitions rather than invocation-tail samples.}
  \label{fig:resident-horizon}
\end{figure*}

%% file: sections/08_resident_results.tex
\section{Mechanism results}
\label{sec:resident-results}

All \ResidentRows{} measured rows pass the frozen status, timing, duration,
source, provider, and correctness gates. Across the two admissible mechanisms,
\ResidentLegalInvocations{} tested batched invocations are field-exact and
decision-exact against the host oracle.

The device-resident path has a lower within-placement median than the host
round trip in all \ResidentCells{} placement-cells. The ratios range from
\ResidentRatioMin$\times$ to \ResidentRatioMax$\times$. At the primary
$N=256,H=32$ cell, the named-placement ratios are 1.71$\times$ on the local
GTX 1660 Ti, 2.39$\times$ on the Modal L4, 2.06$\times$ on the RunPod L4, and
1.84$\times$ on the Lambda H100.

\begin{table*}[t]
  \centering
  \footnotesize
  \input{generated/resident-primary-table.tex}
  \caption{Primary resident-policy cell, $N=256,H=32$. Times are medians of
  batch-average cohort-horizon wall time. The floor omits runtime predicate and
  selection work and is not a legal online policy.}
  \label{tab:resident-primary}
\end{table*}

At that cell, the device path takes \ResidentPrimaryUsMin{} to
\ResidentPrimaryUsMax{} $\mu$s across the four named placements, compared with
\HostPrimaryUsMin{} to \HostPrimaryUsMax{} $\mu$s for the host round trip. The
absolute difference is \PrimarySavedUsMin{} to \PrimarySavedUsMax{} $\mu$s per
32-epoch cohort invocation. The device path is still
\PrimaryFloorRatioMin$\times$ to \PrimaryFloorRatioMax$\times$ slower than the
oracle floor.
This undecomposed gap combines predicate, selector, and graph overhead, which
remains large beside the synthetic route bodies.

\begin{figure*}[t]
  \centering
  \includegraphics[width=0.82\linewidth]{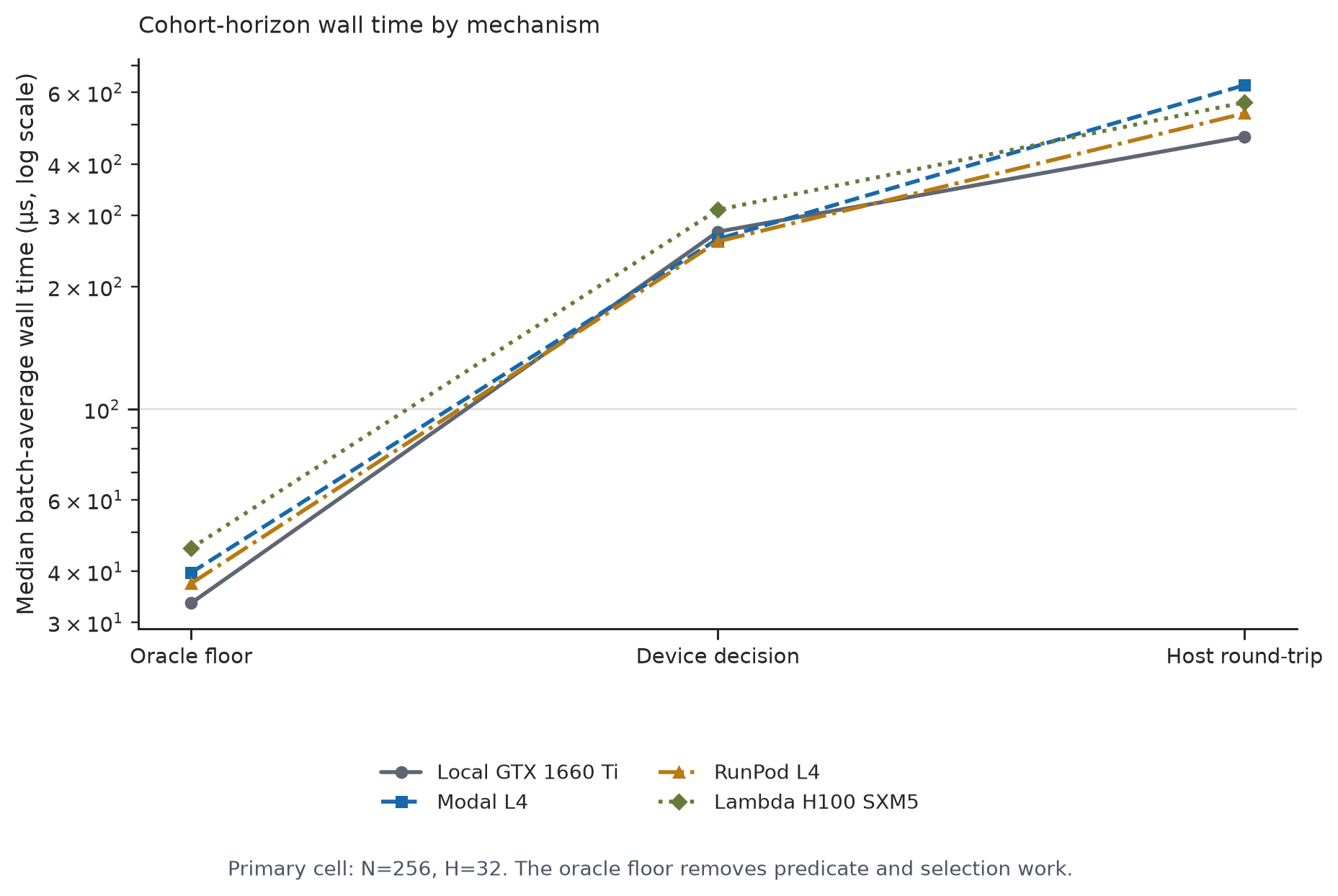}
  \caption{The three mechanisms at the primary cell. Device residency removes
  a matched host observation and dispatch epoch. The oracle floor removes the
  decision itself, so it is diagnostic rather than deployable.}
  \label{fig:resident-primary}
\end{figure*}

\subsection{Negative device-launch control}

The earlier native-dispatch calibration compares a pre-instantiated graph
launched from the host with a fixed child graph launched from a GPU kernel. The
child path retains the host root launch, adds a launcher kernel, and removes no
host decision. Across \NativePlacements{} named placements and all \NativeCells{}
cells, it is slower. The ratio of within-placement wall-time medians ranges
from \NativeRatioMin$\times$ to \NativeRatioMax$\times$. All \NativeRows{} rows
are field-exact.

This negative control separates device residency from device launch. Adding a
nested launcher while removing no host decision produces overhead rather than
an alternative positive treatment.

The timed comparison isolates one control-path choice. Cold graph construction,
ingress, effect egress, validation, and a tuned CPU implementation of a real
route body enter the joined runtime evaluation in
Section~\ref{sec:interpretation}.

\FloatBarrier

%% file: generated/resident-primary-table.tex
\begin{tabular}{@{}llrrrr@{}}
  \toprule
  Provider & GPU & Ratio & Resident ($\mu$s) & Host ($\mu$s) & Floor ($\mu$s) \\
  \midrule
  Local & GTX 1660 Ti & 1.71$\times$ & 273.0 & 467.2 & 33.4 \\
  Modal & L4 & 2.39$\times$ & 261.9 & 625.3 & 39.7 \\
  RunPod & L4 & 2.06$\times$ & 258.3 & 533.3 & 37.3 \\
  Lambda & H100 SXM5 & 1.84$\times$ & 309.0 & 567.5 & 45.6 \\
  \bottomrule
\end{tabular}

%% file: sections/09_joint_interpretation.tex
\section{Joining the evidence}
\label{sec:interpretation}

The two studies define a staged systems test. Cohort supply determines whether
any legal scheduler can find enough same-group work before launch deadlines.
Observation placement determines whether an already resident decision chain
pays a host copy, synchronization, branch, and redispatch epoch. A proposed GPU
route service must pass both gates. Failure at either gate is a concrete reason
to keep that transition on the CPU.

The current numbers do not form a service-level acceleration estimate and must
not be multiplied.
The trace study sweeps candidate thresholds, while the resident study fixes
three cohort sizes and a synthetic route body. The trace also contains no
measurement that 32 consecutive control epochs can execute without a model,
tool, policy, or commit observation. The value of the separation is that each
missing quantity now has a named measurement rather than being hidden inside a
broad GPU-agent claim.

\subsection{The next systems test}

The runtime implied by the boundary is a finite-capacity online route compactor.
It receives typed completion events, updates compact state, queues events by a
verified executable route and deadline, launches a route only above its
measured safe-suffix threshold $K_r$, and falls back to a tuned CPU path when a
deadline or capacity limit prevents batching. The GPU may emit ordered effect
descriptors; a privileged CPU or DPU remains the authority for external
effects.

Such a system adds one essential quantity. Let $z_i=1$ only when the frozen
online policy accelerates event $i$ with an exact output and observed device
start no later than its launch deadline; otherwise $z_i=0$. Then

\begin{equation}
  A=\frac{1}{|E|}\sum_{i\in E}z_i.
\end{equation}

Late, missing, failed, and fallback events have $z_i=0$ and remain in $E$. The
opportunity recovery above fixed windows is

\begin{equation}
  R_A=\frac{A-F}{\pstar-F},
  \label{eq:online-recovery}
\end{equation}

when $\pstar>F$. This turns the offline result into a direct runtime target. A
credible evaluation pairs $A$ and $R_A$ with raw invocation P99, deadline
misses, CPU core-seconds per event, exact trajectories, task utility, network
bytes, cost, and, for a shared inference GPU, TTFT and TPOT guardrails.

\subsection{Deployment paths}

The strongest architecture is a regional route service that aggregates typed
events from many inference and tool workers. Aggregation attacks the cohort
supply problem directly; measured ingress, state movement, and effect egress
then determine whether the control-path savings survive the network boundary.
A second path uses a low-priority queue on an inference GPU that is already
allocated. Its economic advantage is marginal device cost, while its decisive
constraint is noninterference with model TTFT and TPOT.

A dedicated control GPU is the weakest starting point because the present
experiments do not establish enough utilization or CPU displacement to pay for
one. GPU-initiated VM lifecycle operations are also outside the transition
unit: a device can emit a typed prewarm request, but a privileged CPU or DPU
control plane must authorize and commit it.

%% file: sections/10_related_work.tex
\section{Related work}
\label{sec:related}

\paragraph{GPU state machines and device-side control.}
FLAME GPU 2 maps agent memory, state transitions, population partitioning, and
communication to GPUs \cite{richmond2023flamegpu2}. CUDA Graphs provide replay,
conditional nodes, and device graph launch \cite{nvidia2026cudaprogrammingguide}.
GPUOS uses a persistent GPU worker kernel and atomic queues fed from a
host-managed work queue \cite{yang2026gpuos}; MPK and
Event Tensor compile tensor task graphs into persistent megakernels
\cite{cheng2026mpk,jin2026eventtensor}. These systems provide mechanisms for
device-side control and launch amortization. Our measurements ask a different
question: whether agent traces supply enough work under a declared grouping
before launch deadlines, and what one matched host observation costs.

\paragraph{Agent and compound-AI serving.}
A production and open-source characterization by Yang et al. identifies
repeated CPU--GPU crossings, host orchestration on the critical path, and
bursty CPU demand in agentic workflows; its Agora prototype pools and schedules
CPU resources while oversubscribing GPU memory
\cite{yang2026agenticarchitecture}. ThunderAgent represents workflows as LLM
Programs and jointly manages KV caches, system state, and tool environments
with a program-aware scheduler \cite{kang2026thunderagent}.
Parrot exposes semantic variables and application dataflow for cross-request
optimization \cite{lin2024parrot}. Agentix schedules agent programs using
execution progress \cite{luo2026agentix}; SAGA uses KV-cache-aware execution
graphs and session-affinity batching \cite{guo2026saga}; and InferCept manages
inference state across pauses for external interactions
\cite{abhyankar2024infercept}. Murakkab jointly maps declarative workflow
components to models and hardware under SLOs \cite{chaudhry2026murakkab}. The
MARS preprint coordinates GPU inference and CPU tool execution through an
external control plane \cite{wang2026mars}; the concurrent OpRAG preprint uses
persistent workers, bounded queues, and batched GPU operators for multi-stage
RAG \cite{sarker2026oprag}.
AgentServe co-schedules multi-agent inference on one GPU
\cite{zhang2026agentserve}. These systems operate at model-call, tool, or
workflow granularity. Our unit is the deterministic transition after a model or
tool event has completed.

\paragraph{Non-model CPU and GPU placement.}
Agentic CPU-GPU Scheduling profiles a complete AI tool and selects immediate
GPU, queued GPU, or CPU execution \cite{lu2026agenticcpugpu}. TokTier moves
stateful tokenization across CPU and GPU while enforcing exact token identity
\cite{zhang2026toktier}. Measurements of CPU-induced LLM slowdowns show that
CPU starvation can delay launches, communication, and tokenization even with
CUDA Graphs, and that adding CPU cores can be the economical remedy
\cite{chung2026cpuslowdowns}. These are the closest tool-level and methodological
comparisons. They make a tuned CPU implementation and CPU provisioning
mandatory baselines for the proposed online runtime.

\paragraph{Batching and clustering.}
Classical real-time scheduling studies job families, time windows, and
compatible batching \cite{barnoy2009throughput}. Joint batching work addresses
heterogeneous asynchronous arrivals and deadlines \cite{cang2024jointbatching},
while serial-batch scheduling includes minimum batch size, job families, and
release times \cite{huertas2025serialbatch}. One-dimensional $r$-gathering
studies minimum-occupancy clusters, outliers, contiguous solutions, and exact
line algorithms \cite{aggarwal2010anonymity,akagi2015rgather,
nakano2019rgather,sarker2021rgather}. Univariate microaggregation with
suppression supplies another dynamic-programming analogue
\cite{laszlo2013microaggregation}. GPU dynamic batching has also been modeled
with batch-dependent service, stochastic arrivals, latency, and power
\cite{xu2023smdpbatching}. The recurrence in this paper is a restricted trace
oracle inside that established territory, not an algorithmic-priority claim.

\paragraph{Positioning.}
\claimid{RC-012}
Prior agent systems characterize serving, scheduling, state management, and
tool execution, while batching theory characterizes which arrivals can be
served together under timing and occupancy constraints. The supply of
route-key-conditioned post-event transitions above a hardware crossover lies
at their interface. Observation placement adds the second axis: a cohort may
exist, yet a host round trip can still dominate a short control chain. This
paper gives those axes common boundary quantities and measures each with a
separate, auditable experiment.

%% file: sections/11_limitations.tex
\section{Scope of inference}
\label{sec:limitations}

\paragraph{Model and evidence join.}
The trace threshold is a swept candidate rather than the measured crossover of
the resident-policy source, and the mechanism's $H=32$ horizon is absent from
the trace model. The two result sets therefore cannot be multiplied. The
offline optimum also assumes future knowledge, zero service time, unlimited
capacity, no upper batch size, and deadlines on launch rather than completion.
The safe-suffix abstraction extrapolates beyond tested cohort sizes and does
not cover non-monotone benefit, arbitrary deadlines, or batch-size-dependent
service.

\paragraph{Trace scope and executable identity.}
The replay conditions on one \TraceSessions-session panel from three related
customer-service domains and imposes stationary Poisson arrivals. Its three
seeds measure Monte Carlo variation under that panel and model, not a workload
population. Bursts, correlated releases, and other corpora can move the
boundary. The outcome-derived route key omits state-machine node, schema,
arguments, policy context, and multi-tool identities, so it is a conditioning
proxy rather than verified semantic fusion. The packing arrays also omit a
per-session sequence constraint; overlapping spans and completion-order
inversions can affect the wider deadline surface.

\paragraph{Mechanism and sampling scope.}
Resident-policy-001 makes one global binary decision over a regular synthetic
state array. A route service still needs per-event compaction, variable route
bodies, ingress and egress, effect ordering, recovery, and CPU fallback. Setup,
state reset, result copy, and validation are outside timing. The four named
placements comprise two L4s, one H100, and one GTX; GPU, provider, host, image,
driver, and region are confounded. The local full run reuses the development
GPU, and one external placement extends the frozen first-stage scope. Timing
rows are technical repetitions, so the current performance result supports
named-placement effect sizes rather than a hardware-population inference.

\paragraph{Deployment and algorithmic scope.}
The mechanism metric is batch-average cohort-horizon wall time, not invocation
P99 or end-to-end task time. A tuned CPU implementation of the real transition,
CPU core-seconds, energy, cost, model throughput, TTFT, TPOT, task utility, and
external-effect reliability remain deployment measurements. The dynamic
program is a specialized exact instrument within mature batching, clustering,
and scheduling literatures; the paper makes no exhaustive algorithmic priority
claim.

%% file: sections/12_reproducibility.tex
\section{Reproducibility}
\label{sec:reproducibility}

The public code and release materials are at
\url{https://github.com/josefchen/ready-cohorts}. A processed-evidence mirror is
at \url{https://huggingface.co/datasets/josefchen/ready-cohorts}.

Raw measurements, processed tables, analysis code, preregistrations, and paper
outputs are separate artifacts. New executions receive new experiment and
placement identities; analysis does not overwrite a prior raw run.

All numerical macros and primary tables in the manuscript are generated by
\texttt{scripts/build\_paper\_artifacts.py}. The generator validates declared
hashes for the trace summary, repetition rows, and source dependencies,
including all 19 local Parquet shards; it also validates the resident-policy
contrasts and cell summary, and the native-dispatch contrasts. It checks the
primary cell, selected row and placement counts, the pointwise boundary
invariant, recorded trace-gate flags, and directional inequalities before
emitting a machine-readable paper-data manifest. Session and span counts come
from the source manifest rather than manuscript constants. This build validates
retained and processed evidence; it does not rerun cloud experiments or repeat
the remote retrieval.

The source manifest binds the trace extraction to dataset commit
\sha{70036b93a04e61b0ea2706a68b962f4f26774587} and Parquet conversion commit
\sha{f7c94012d0bfbf66fe4d6ed627699508bbb555ff}. The SHA-256 values of all 19
local shards match both the manifest and the commit-resolved remote bytes.
Derived features retain timestamps, counts, lengths, route labels, and public
identifiers, but no prompt text, tool arguments, or tool results.

The retained figures are hash-bound to their source table, notebook, and
notebook-builder script. They are not regenerated by the LaTeX target. The
claim-evidence map links the headline claims to proofs or artifacts. Placement
is the performance sampling unit; timing rows and batched invocations remain
technical repetitions.

For the packing result, the artifact separates the recurrence's attainable
$O(N\log N)$ grouped implementation from the frozen evaluator used here. The
frozen code rescans the full grouping array once per route and performs a binary
boundary search per event, for
$O(NR+\sum_r n_r\log n_r)$ time in the stated accounting. Integer nanosecond
normalization and brute-force tiny-instance tests define its exactness contract.

From the repository root, the manuscript is rebuilt and checked with:

\begin{verbatim}
make -C paper/arxiv clean all
.venv/bin/python scripts/check_arxiv_paper.py
\end{verbatim}

\paragraph{Generative AI disclosure.}
OpenAI Codex assisted with implementation, experiment orchestration, literature
search, quantitative checks, adversarial review, and language editing. Reported
numbers are regenerated from retained artifacts, and citations are checked
against primary sources. The author remains responsible for the scientific
claims and the released artifact.

%% file: sections/13_conclusion.tex
\section{Conclusion}
\label{sec:conclusion}

The ready-cohort boundary turns GPU agent control into a falsifiable systems
question. A workload must supply enough same-group events before their launch
deadlines, and the control path must account for where intermediate decisions
are observed. Under the frozen trace model, exact sliding-deadline packing
raises eligible share from \TraceFixedPct\% to \TraceExactPct\%, recovering
\TraceGapClosurePct\% of the opportunity lost by fixed windows. The full
surface also shows where cohort supply collapses to zero.

The mechanism study establishes the second gate. Keeping the tested binary
decision on device reduces cohort-horizon wall time in all
\ResidentCells{} cells across four named placements, with ratios from
\ResidentRatioMin$\times$ to \ResidentRatioMax$\times$. A nested device graph
that removes no host decision loses in all \NativeCells{} cells across five
placements. Device launch alone is therefore insufficient; the gain appears
only in the treatment that removes the matched host-mediated decision epoch.

These results provide the workload budget and the mechanism test for an online
route compactor. Its decisive measurements are achieved share $A$ relative to
$F$ and $\pstar$, raw P99, CPU core-seconds, exact effects, task utility, cost,
and shared-inference guardrails. If it cannot convert offline headroom into
online work or displace CPU without harming model service, the boundary rejects
the GPU design. If it can, the same quantities show exactly where GPU-resident
agent control belongs in a datacenter.

%% file: sections/appendix.tex
\section{Proof and evaluator scope}
\label{app:proof-scope}

\subsection{Binary-program reference formulation}

For arbitrary deadlines under the same zero-service and unlimited-capacity
assumptions, candidate launch times may be restricted to unique deadlines by
Proposition~\ref{prop:deadlines}. Let $x_{i\tau}$ assign event $i$ to deadline $\tau$,
let $y_{r\tau}$ indicate a batch for route $r$, and let $M_{r\tau}$ be the
number of route-$r$ event intervals containing $\tau$. One reference
formulation is

\begin{align}
\max\quad &\sum_i\sum_\tau x_{i\tau} \\
\text{s.t.}\quad
&\sum_\tau x_{i\tau}\leq1 &&\forall i,\\
&x_{i\tau}=0 &&\text{if }\tau\notin[t_i,d_i],\\
&K_r y_{r\tau}\leq\sum_{i:r_i=r}x_{i\tau} &&\forall r,\tau,\\
&\sum_{i:r_i=r}x_{i\tau}\leq M_{r\tau}y_{r\tau} &&\forall r,\tau,\\
&x_{i\tau},y_{r\tau}\in\{0,1\}.
\end{align}

The objective does not minimize batch count or wait among
maximum-cardinality schedules. The reported equal-deadline evaluator avoids
this quadratic candidate assignment and returns one witness.

\section{Evidence inventory}
\label{app:evidence}

\begin{table}[htbp]
  \centering
  \footnotesize
  \begin{tabularx}{\linewidth}{@{}p{0.19\linewidth} X p{0.27\linewidth}@{}}
    \toprule
    Layer & Current evidence and outer unit & Admissible inference \\
    \midrule
    Route-key trace packing & Replay-seed unit conditional on one panel;
      540 cell-seed rows from nine swarms; all invariants pass & Conditional opportunity surface \\
    Primary trace cell & Same outer unit;
      $F=\TraceFixedPct\%$, $\pstar=\TraceExactPct\%$, $U=\TraceUpperPct\%$ &
      Exact offline share under the model \\
    Resident decision & Named-placement unit; four placements and
      \ResidentCells{} directional cells & Named-placement effect direction and size \\
    Exactness stress & Tested-batched-invocation unit;
      \ResidentLegalInvocations{} admissible-mechanism invocations & Exactness on tested invocations \\
    Nested launch & Named-placement unit; five placements and
      \NativeCells{} negative cells & Fixed nested-launch calibration \\
    \bottomrule
  \end{tabularx}
  \caption{Evidence layers, outer units, and admissible inference. Population
  and end-to-end claims require the joined confirmation design.}
  \label{tab:evidence-inventory}
\end{table}

\section{Protocol deviations}
\label{app:deviations}

\paragraph{Resident policy.}

\begin{enumerate}
  \item A correctness and engineering smoke preceded the source freeze. The
  smoke and full local run used the same GPU UUID and are not independent
  placements.
  \item The cloud protocol authorized Modal plus one of RunPod or Lambda. Both
  external placements were ultimately retained. The additional placement is a
  disclosed descriptive scope expansion.
  \item Horizon-ratio monotonicity is exploratory because the preregistration
  did not fully specify its operational test.
  \item No unsupported account of failed Lambda provisioning attempts is used
  as scientific evidence. Retained receipts support the successful execution
  and final resource absence.
\end{enumerate}

\paragraph{Native dispatch.}
\label{app:native-deviation}

The native-dispatch plan proposed two fresh placements per available
provider/SKU and at least six H100 placements for nuisance estimation. That
layout was not completed. The retained calibration contains five named
placements, including one H100, three L4s, and one GTX, and supports only the
descriptive negative result reported here.

%% file: main.bbl
\begin{thebibliography}{30}
\providecommand{\natexlab}[1]{#1}
\providecommand{\url}[1]{\texttt{#1}}
\expandafter\ifx\csname urlstyle\endcsname\relax
  \providecommand{\doi}[1]{doi: #1}\else
  \providecommand{\doi}{doi: \begingroup \urlstyle{rm}\Url}\fi

\bibitem[Abhyankar et~al.(2024)Abhyankar, He, Srivatsa, Zhang, and
  Zhang]{abhyankar2024infercept}
Reyna Abhyankar, Zijian He, Vikranth Srivatsa, Hao Zhang, and Yiying Zhang.
\newblock Infercept: Efficient intercept support for augmented large language
  model inference.
\newblock In \emph{Proceedings of the 41st International Conference on Machine
  Learning}, volume 235 of \emph{Proceedings of Machine Learning Research},
  pages 81--95. PMLR, 2024.
\newblock URL \url{https://proceedings.mlr.press/v235/abhyankar24a.html}.

\bibitem[Aggarwal et~al.(2010)Aggarwal, Feder, Kenthapadi, Khuller, Panigrahy,
  Thomas, and Zhu]{aggarwal2010anonymity}
Gagan Aggarwal, Tom{\'a}s Feder, Krishnaram Kenthapadi, Samir Khuller, Rina
  Panigrahy, Dilys Thomas, and An~Zhu.
\newblock Achieving anonymity via clustering.
\newblock \emph{ACM Transactions on Algorithms}, 6\penalty0 (3):\penalty0
  1--19, 2010.
\newblock \doi{10.1145/1798596.1798602}.
\newblock URL \url{https://doi.org/10.1145/1798596.1798602}.

\bibitem[Akagi and ichi Nakano(2015)]{akagi2015rgather}
Toshihiro Akagi and Shin ichi Nakano.
\newblock On {$r$}-gatherings on the line.
\newblock In \emph{Frontiers in Algorithmics}, volume 9130 of \emph{Lecture
  Notes in Computer Science}, pages 25--32. Springer, 2015.
\newblock \doi{10.1007/978-3-319-19647-3_3}.
\newblock URL \url{https://doi.org/10.1007/978-3-319-19647-3_3}.

\bibitem[Bar-Noy et~al.(2009)Bar-Noy, Guha, Katz, Naor, Schieber, and
  Shachnai]{barnoy2009throughput}
Amotz Bar-Noy, Sudipto Guha, Yoav Katz, Joseph Naor, Baruch Schieber, and Hadas
  Shachnai.
\newblock Throughput maximization of real-time scheduling with batching.
\newblock \emph{ACM Transactions on Algorithms}, 5\penalty0 (2):\penalty0
  18:1--18:17, 2009.
\newblock \doi{10.1145/1497290.1497294}.
\newblock URL \url{https://doi.org/10.1145/1497290.1497294}.

\bibitem[Barres et~al.(2026)Barres, Dong, Ray, Si, and
  Narasimhan]{barres2026tau2bench}
Victor Barres, Honghua Dong, Soham Ray, Xujie Si, and Karthik Narasimhan.
\newblock {$\tau^2$-Bench}: Evaluating conversational agents in a dual-control
  environment.
\newblock In \emph{Proceedings of the 43rd International Conference on Machine
  Learning}, volume 306 of \emph{Proceedings of Machine Learning Research}.
  PMLR, 2026.
\newblock URL \url{https://openreview.net/forum?id=OC2z7iSQKa}.

\bibitem[Cang et~al.(2024)Cang, Chen, and Huang]{cang2024jointbatching}
Yihan Cang, Ming Chen, and Kaibin Huang.
\newblock Joint batching and scheduling for high-throughput multiuser edge {AI}
  with asynchronous task arrivals.
\newblock \emph{IEEE Transactions on Wireless Communications}, 23\penalty0
  (10):\penalty0 13782--13795, 2024.
\newblock \doi{10.1109/TWC.2024.3404811}.
\newblock URL \url{https://doi.org/10.1109/TWC.2024.3404811}.

\bibitem[Chaudhry et~al.(2026)Chaudhry, Choukse, Qiu, Goiri, Fonseca, Belay,
  and Bianchini]{chaudhry2026murakkab}
Gohar~Irfan Chaudhry, Esha Choukse, Haoran Qiu, Inigo Goiri, Rodrigo Fonseca,
  Adam Belay, and Ricardo Bianchini.
\newblock Murakkab: Resource-efficient agentic workflow orchestration in cloud
  platforms.
\newblock In \emph{20th USENIX Symposium on Operating Systems Design and
  Implementation (OSDI 26)}, pages 567--587, Seattle, WA, July 2026. USENIX
  Association.
\newblock URL
  \url{https://www.usenix.org/conference/osdi26/presentation/chaudhry}.

\bibitem[Cheng et~al.(2026)Cheng, Zhang, Zhou, Ji, Jiang, Zhao, Xiao, Ye,
  Huang, Lai, Jin, Hou, Wu, Dong, Yip, Wang, Yang, Miao, Chen, and
  Jia]{cheng2026mpk}
Xinhao Cheng, Zhihao Zhang, Yu~Zhou, Jianan Ji, Jinchen Jiang, Zepeng Zhao,
  Ziruo Xiao, Zihao Ye, Yingyi Huang, Ruihang Lai, Hongyi Jin, Bohan Hou,
  Mengdi Wu, Yixin Dong, Anthony Yip, Songting Wang, Wenqin Yang, Xupeng Miao,
  Tianqi Chen, and Zhihao Jia.
\newblock {MPK}: A compiler and runtime for {Mega-Kernelizing} tensor programs.
\newblock In \emph{20th USENIX Symposium on Operating Systems Design and
  Implementation (OSDI 26)}, pages 1909--1926, Seattle, WA, July 2026. USENIX
  Association.
\newblock ISBN 978-1-939133-55-7.
\newblock URL
  \url{https://www.usenix.org/conference/osdi26/presentation/cheng}.

\bibitem[Chung et~al.(2026)Chung, Jia, Jezghani, and
  Kim]{chung2026cpuslowdowns}
Euijun Chung, Yuxiao Jia, Aaron Jezghani, and Hyesoon Kim.
\newblock Characterizing {CPU}-induced slowdowns in multi-{GPU} {LLM}
  inference, 2026.
\newblock URL \url{https://arxiv.org/abs/2603.22774}.

\bibitem[{Exgentic}(2026)]{exgentic2026traces}
{Exgentic}.
\newblock Multi-benchmark {LLM} agent traces.
\newblock Hugging Face dataset, 2026.
\newblock URL
  \url{https://huggingface.co/datasets/Exgentic/agent-llm-traces/tree/f7c94012d0bfbf66fe4d6ed627699508bbb555ff}.
\newblock Revision 70036b93a04e61b0ea2706a68b962f4f26774587;
  CDLA-Permissive-2.0; accessed 2026-08-12.

\bibitem[Guo et~al.(2026)Guo, Wu, and Yiu]{guo2026saga}
Dongxin Guo, Jikun Wu, and Siu~Ming Yiu.
\newblock {SAGA}: Workflow-atomic scheduling for {AI} agent inference on {GPU}
  clusters, 2026.
\newblock URL \url{https://arxiv.org/abs/2605.00528}.

\bibitem[Huertas and Hentenryck(2025)]{huertas2025serialbatch}
Jorge~A. Huertas and Pascal~Van Hentenryck.
\newblock Constraint programming models for serial batch scheduling with
  minimum batch size.
\newblock \emph{Operations Research Perspectives}, 15:\penalty0 100352, 2025.
\newblock \doi{10.1016/j.orp.2025.100352}.
\newblock URL \url{https://doi.org/10.1016/j.orp.2025.100352}.

\bibitem[ichi Nakano(2019)]{nakano2019rgather}
Shin ichi Nakano.
\newblock A simple algorithm for {$r$}-gatherings on the line.
\newblock \emph{Journal of Graph Algorithms and Applications}, 23\penalty0
  (5):\penalty0 837--845, 2019.
\newblock \doi{10.7155/jgaa.00514}.
\newblock URL \url{https://doi.org/10.7155/jgaa.00514}.

\bibitem[Jin et~al.(2026)Jin, Hou, Wang, Lai, Chen, Ye, Cai, Dong, Cheng,
  Zhang, Zhao, Huang, Yang, Jiang, Oliaro, Ji, Miao, Grover, Mowry, Jia, and
  Chen]{jin2026eventtensor}
Hongyi Jin, Bohan Hou, Guanjie Wang, Ruihang Lai, Jinqi Chen, Zihao Ye, Yaxing
  Cai, Yixin Dong, Xinhao Cheng, Zhihao Zhang, Yilong Zhao, Yingyi Huang, Lijie
  Yang, Jinchen Jiang, Gabriele Oliaro, Jianan Ji, Xupeng Miao, Vinod Grover,
  Todd~C. Mowry, Zhihao Jia, and Tianqi Chen.
\newblock Event tensor: A unified abstraction for compiling dynamic megakernel.
\newblock In \emph{Proceedings of Machine Learning and Systems}, volume~8,
  pages 1917--1933. MLSys, 2026.
\newblock URL
  \url{https://proceedings.mlsys.org/paper_files/paper/2026/file/53d3f45797970d323bd8a0d379c525aa-Paper-Conference.pdf}.

\bibitem[Kang et~al.(2026)Kang, Li, Xu, Yang, Chen, Wang, Chen, Krishna, Xu,
  and Arora]{kang2026thunderagent}
Hao Kang, Ziyang Li, Weili Xu, Xinyu Yang, Yinfang Chen, Junxiong Wang, Beidi
  Chen, Tushar Krishna, Chenfeng Xu, and Simran Arora.
\newblock {ThunderAgent}: A simple, fast and program-aware agentic inference
  system, 2026.
\newblock URL \url{https://arxiv.org/abs/2602.13692}.

\bibitem[Laszlo and Mukherjee(2013)]{laszlo2013microaggregation}
Michael~J. Laszlo and Sumitra Mukherjee.
\newblock Optimal univariate microaggregation with data suppression.
\newblock \emph{Journal of Systems and Software}, 86\penalty0 (3):\penalty0
  677--682, 2013.
\newblock \doi{10.1016/j.jss.2012.10.901}.
\newblock URL \url{https://doi.org/10.1016/j.jss.2012.10.901}.

\bibitem[Lin et~al.(2024)Lin, Han, Zhang, Yang, Yang, Chen, and
  Qiu]{lin2024parrot}
Chaofan Lin, Zhenhua Han, Chengruidong Zhang, Yuqing Yang, Fan Yang, Chen Chen,
  and Lili Qiu.
\newblock Parrot: Efficient serving of {LLM-based} applications with semantic
  variable.
\newblock In \emph{18th USENIX Symposium on Operating Systems Design and
  Implementation (OSDI 24)}, pages 929--945, Santa Clara, CA, July 2024. USENIX
  Association.
\newblock ISBN 978-1-939133-40-3.
\newblock URL
  \url{https://www.usenix.org/conference/osdi24/presentation/lin-chaofan}.

\bibitem[Lu and Reda(2026)]{lu2026agenticcpugpu}
Tianxi Lu and Sherief Reda.
\newblock Agentic {CPU-GPU} scheduling for heterogeneous {AI} workloads, 2026.
\newblock URL \url{https://arxiv.org/abs/2607.22242}.

\bibitem[Luo et~al.(2026)Luo, Shi, Cai, Zhang, Wong, Wang, Wang, Huang, Chen,
  Gonzalez, and Stoica]{luo2026agentix}
Michael Luo, Xiaoxiang Shi, Colin Cai, Tianjun Zhang, Justin Wong, Yichuan
  Wang, Chi Wang, Yanping Huang, Zhifeng Chen, Joseph~E. Gonzalez, and Ion
  Stoica.
\newblock Agentix: An efficient serving engine for {LLM} agents as general
  programs.
\newblock In \emph{23rd USENIX Symposium on Networked Systems Design and
  Implementation (NSDI 26)}, pages 2443--2459, Renton, WA, May 2026. USENIX
  Association.
\newblock ISBN 978-1-939133-54-0.
\newblock URL \url{https://www.usenix.org/conference/nsdi26/presentation/luo}.

\bibitem[NVI(2026)]{nvidia2026cudaprogrammingguide}
\emph{{CUDA Programming Guide}}.
\newblock NVIDIA Corporation, 2026.
\newblock URL
  \url{https://docs.nvidia.com/cuda/cuda-programming-guide/04-special-topics/cuda-graphs.html}.
\newblock Chapter 4.2, CUDA Graphs; accessed 2026-08-12.

\bibitem[Richmond et~al.(2023)Richmond, Chisholm, Heywood, Chimeh, and
  Leach]{richmond2023flamegpu2}
Paul Richmond, Robert Chisholm, Peter Heywood, Mozhgan~Kabiri Chimeh, and
  Matthew Leach.
\newblock {FLAME GPU 2}: A framework for flexible and performant agent based
  simulation on {GPU}s.
\newblock \emph{Software: Practice and Experience}, 53\penalty0 (8):\penalty0
  1659--1680, 2023.
\newblock \doi{10.1002/spe.3207}.
\newblock URL \url{https://doi.org/10.1002/spe.3207}.

\bibitem[Sarker et~al.(2021)Sarker, Sung, and Rahman]{sarker2021rgather}
Anik Sarker, Wing-Kin Sung, and M.~Sohel Rahman.
\newblock A linear time algorithm for the {$r$}-gathering problem on the line.
\newblock \emph{Theoretical Computer Science}, 866:\penalty0 96--106, 2021.
\newblock \doi{10.1016/j.tcs.2021.03.015}.
\newblock URL \url{https://doi.org/10.1016/j.tcs.2021.03.015}.

\bibitem[Sarker et~al.(2026)Sarker, Staylor, Alsaadi, von Laszewski, Jha, and
  Fox]{sarker2026oprag}
Arup~Kumar Sarker, Mills Staylor, Aymen Alsaadi, Gregor von Laszewski, Shantenu
  Jha, and Geoffrey Fox.
\newblock {OpRAG}: A resource-deterministic runtime for gpu-backed multi-stage
  {RAG} workflows, 2026.
\newblock URL \url{https://arxiv.org/abs/2608.08340}.

\bibitem[Wang et~al.(2026)Wang, Ye, Xu, Guo, Wei, Wang, Li, Chen, Li, Zhuo, and
  Chen]{wang2026mars}
Yifei Wang, Hancheng Ye, Yechen Xu, Cong Guo, Chiyue Wei, Qinsi Wang, Dongting
  Li, Tingjun Chen, Hai Li, Danyang Zhuo, and Yiran Chen.
\newblock {MARS}: Efficient, adaptive co-scheduling for heterogeneous agentic
  systems, 2026.
\newblock URL \url{https://arxiv.org/abs/2604.26963}.

\bibitem[Xu et~al.(2023)Xu, Sun, Zhou, and Niu]{xu2023smdpbatching}
Yaodan Xu, Jingzhou Sun, Sheng Zhou, and Zhisheng Niu.
\newblock {SMDP}-based dynamic batching for efficient inference on {GPU}-based
  platforms.
\newblock In \emph{2023 IEEE International Conference on Communications (ICC)},
  pages 5483--5489. IEEE, 2023.
\newblock \doi{10.1109/ICC45041.2023.10278962}.
\newblock URL \url{https://doi.org/10.1109/ICC45041.2023.10278962}.

\bibitem[Yang et~al.(2026{\natexlab{a}})Yang, Liu, Zhang, and
  Stojkovic]{yang2026agenticarchitecture}
Jirong Yang, Peizhe Liu, Chaojie Zhang, and Jovan Stojkovic.
\newblock Architectural implications of agentic {AI} workflows,
  2026{\natexlab{a}}.
\newblock URL \url{https://arxiv.org/abs/2608.04458}.

\bibitem[Yang et~al.(2026{\natexlab{b}})Yang, Gao, Zhou, Gan, Zheng, and
  Quinn]{yang2026gpuos}
Yiwei Yang, Xiangyu Gao, Yuan Zhou, Yuhang Gan, Yusheng Zheng, and Andi Quinn.
\newblock {GPUOS}: A {GPU} operating system primitive for transparent operation
  fusion, 2026{\natexlab{b}}.
\newblock URL \url{https://arxiv.org/abs/2604.17861}.

\bibitem[Yao et~al.(2024)Yao, Shinn, Razavi, and Narasimhan]{yao2024taubench}
Shunyu Yao, Noah Shinn, Pedram Razavi, and Karthik Narasimhan.
\newblock {$\tau$-bench}: A benchmark for tool-agent-user interaction in
  real-world domains, 2024.
\newblock URL \url{https://arxiv.org/abs/2406.12045}.

\bibitem[Zhang et~al.(2026)Zhang, Yan, Yang, and Yuan]{zhang2026agentserve}
Yuning Zhang, Yan Yan, Nan Yang, and Dong Yuan.
\newblock {AgentServe}: Algorithm-system co-design for efficient agentic {AI}
  serving on a consumer-grade {GPU}, 2026.
\newblock URL \url{https://arxiv.org/abs/2603.10342}.

\bibitem[Zhang and Cao(2026)]{zhang2026toktier}
Zhenyu Zhang and Zhichao Cao.
\newblock {TokTier}: Exact stateful {CPU+GPU} tokenization for agentic {LLM}
  serving, 2026.
\newblock URL \url{https://arxiv.org/abs/2607.29678}.

\end{thebibliography}
